\documentclass[journal,draftcls,onecolumn]{IEEEtran}

\usepackage[utf8]{inputenc} 
\usepackage[T1]{fontenc}
\usepackage{url}
\usepackage{ifthen}
\usepackage[cmex10]{amsmath} 
\usepackage{enumitem}

\usepackage{tikz}
\usepackage{graphicx}
\usetikzlibrary{arrows.meta,calc,positioning,calc}

\usepackage{amsthm}
\newtheoremstyle{mystyle}
  {\topsep} 
  {\topsep} 
  {\normalfont} 
  {\parindent} 
  {\itshape} 
  {:} 
  { } 
  {} 

\theoremstyle{mystyle}

\newtheorem{prop}{Proposition}

\newtheorem{theorem}{Theorem}
\newtheorem{lemma}{Lemma}

\newtheorem{cor}{Corollary}
\newtheorem{remark}{Remark}
\newtheorem{example}{Example}

\usepackage{booktabs}
\usepackage{tabularx}
\usepackage{array}
\usepackage{multirow}

\usepackage{soul}

\usepackage{amsmath,amssymb,amsfonts}

\usepackage{dsfont}
\usepackage{algorithm}
\usepackage{algpseudocode}
\usepackage{graphicx}
\usepackage{textcomp}
\usepackage{xcolor}
\usepackage{braket}
\usepackage{hyperref}
\usepackage{amsthm}
\usepackage{url}

\usepackage{amsmath}

\usepackage{graphicx}
\usepackage{subcaption}

\newcommand{\bfU}{\mathbf{U}}

\newcommand{\bfr}{\mathbf{r}}

\newcommand{\bfD}{\mathbf{D}}

\newcommand{\bfH}{\mathbf{H}}
\newcommand{\bfV}{\mathbf{V}}

\newcommand{\bfT}{\mathbf{T}}
\newcommand{\bfG}{\mathbf{G}}

\newcommand{\bfg}{\mathbf{g}}

\newcommand{\bfB}{\mathbf{B}}

\newcommand{\bfM}{\mathbf{M}}
\newcommand{\bfR}{\mathbf{R}}

\newcommand{\calD}{\mathcal{D}}

\newcommand{\diag}{\text{diag}}

\newcommand{\norm}[1]{\left\lVert #1\right\rVert}

\usepackage{algorithm}

\title{Communication-Efficient Distributed Training via Ring-Based Coded Approximate All-Reduce}

\author{Sifat Munim and Aditya Ramamoorthy%
}

\begin{document}
\maketitle
\begin{abstract}
Ring All-Reduce is a widely used protocol within large-scale distributed training for aggregating gradients across workers in each iteration. For a system with $N$ workers, its normalized per-worker communication is $2(N-1)/N$. In this work we present a communication-efficient Ring All-Reduce (CERAR) protocol for ``approximate'' gradient aggregation. CERAR partitions each local gradient into $c$ components and crucially relies on linear encoding and decoding operations. It performs $L=c+N-2$ communication rounds over an $N$-worker ring, yielding normalized communication rate $1+(N-2)/c$ and normalized storage rate $1+N/c$. We present an explicit Vandermonde matrix based construction, whose approximation error can be made arbitrarily close to zero with communication and storage rates approaching one with increasing $c$. However, this limit is achieved through ill-conditioned encoding matrices. Accordingly, we give an alternate construction whose error is $O(\epsilon)$, while the relevant condition numbers are $O(\epsilon^{-r_\star})$, where $r_\star=\lceil c/N\rceil-1$. This naturally motivates a condition-number-constrained optimization formulation for trading off the competing objectives and obtaining numerically stable practical designs.

Our numerical experiments demonstrate that when compared with the production-standard NCCL on NVIDIA A100 PCIe GPU Clusters with low-bandwidth PCIe interconnect, CERAR reduces the raw All-Reduce phase time by up to 47\%, and for training a 1.4-billion-parameter Pythia model it reduces training time by about 17\% for the same number of iterations while achieving essentially identical validation loss. On the other hand, for GPU clusters with high-bandwidth NVLink interconnects, CERAR performs worse on the All-Reduce phase time and slightly worse in terms of training. Nevertheless, the gap reduces with experiments that involve systems with larger parameter counts. In this scenario we expect the gains to manifest with more optimized implementations of the CERAR protocol. 

\end{abstract}

\section{Introduction}
\label{sec:intro}

A typical scenario within model training involves a dataset $\mathcal{D} = \{(\mathbf{x}_i, y_i)\}_{i = 1}^{\tilde{N}}$ of $\tilde{N}$ data points, where $\mathbf{x}_i$'s and $y_i$'s are the features and labels, respectively. The aim is to minimize a loss function $ L \triangleq \frac{1}{\tilde{N}}\sum_{i = 1}^{\tilde{N}} l(\mathbf{x}_i, y_i; \mathbf{w})$ with respect to $\mathbf{w}$, where $\mathbf{w} \in \mathbb{R}^d$ is the unknown parameter that needs to be learned. The function $l$ measures the prediction error for each data point. In the distributed setting, portions of the dataset are distributed to different workers who compute the gradient on the data points assigned to them. At the $t$-th iteration, the system aims at computing the overall gradient $\nabla L = \frac{1}{\tilde{N}}\sum_{i=1}^{\tilde{N}} \nabla l(\mathbf{x}_i, y_i; \mathbf{w}_t)$, following which a parameter-update is performed and the iterations continue thereafter. We emphasize that $\calD$ can also represent a relevant subset of the entire dataset, e.g., when the cluster wants to run mini-batch stochastic gradient descent \cite{bottou_optimization} on an appropriate batch; this is in fact the most common mode of usage. Several distributed learning architectures have been investigated in the literature  \cite{langer2020distributed}, that differ based on how the dataset is divided into chunks and how the workers are connected to each other.

{\noindent {\bf Background and Motivation:}} For $N$ workers, the basic Ring All-Reduce training algorithm \cite{patarasuk2009bandwidth,sergeev2018horovod} operates by partitioning $\calD$ into $N$ equal-sized disjoint chunks $\mathcal{D}_j, j = 0, \dots, N-1$, such that worker $W_j$ is assigned $\mathcal{D}_j$. $W_j$ computes the sum of the gradients on its assigned data points, i.e., it computes $g_{\calD_j} = \frac{1}{\tilde{N}}\sum_{i \in \mathcal{D}_j}  \nabla l(\mathbf{x}_i, y_i; \mathbf{w}_t)$ (a vector of length-$d$). Once the workers have computed their assigned gradients, they send messages in parallel around a clockwise ring such that each worker can compute the overall gradient $\nabla L$. For a ring with $N$ workers, each worker transmits and receives a total of $2(N-1)d/N$ real values. For large $N$, the load normalized by $d$ approaches $2$, i.e., it does not grow with the number of workers. In contrast, in older architectures such as parameter server (PS) + workers \cite{li2014scaling}, the bandwidth load on the PS grows with $N$.

The basic Ring All-Reduce algorithm is nominally considered to be optimal from the point of view of recovering the sum $g_{\calD_0} + \dots + g_{\calD_{N-1}}$ exactly \cite{patarasuk2009bandwidth}. From an information-theoretic perspective, one can consider the problem of recovering the sum with arbitrarily small (but nonzero) error; the lower bound of \cite{patarasuk2009bandwidth} does not apply in this case. 
We note that approximate recovery often suffices in scenarios such as distributed parameter training, i.e., instead of true gradients, we often settle for approximate gradients which do not affect the quality of training much. We consider this approximate recovery setting in this paper and investigate the fundamental question of how small the communication load on a worker can be in this setting.

\begin{figure}[!t]
\centering
\begin{minipage}[t]{0.49\linewidth}
\centering
\resizebox{\linewidth}{!}{%
\begin{tikzpicture}[
 worker/.style={
   circle, draw=black!75, fill=blue!7,
   line width=0.8pt, minimum size=11mm, font=\large
 },
 arrow/.style={
   -{Stealth[length=2.5mm]},
   line width=0.9pt, blue!65!black
 },
 message/.style={
   font=\large, align=center,
   fill=white, inner sep=3pt
 },
 update/.style={
   font=\large, align=center, text=red!70!black
 }
]
\node[font=\large\bfseries] at (3,6.15) {Round 4};

\node[worker] (w0) at (0,4) {$W_0$};
\node[worker] (w1) at (6,4) {$W_1$};
\node[worker] (w2) at (6,0) {$W_2$};
\node[worker] (w3) at (0,0) {$W_3$};

\draw[arrow] (w0.east) --
  node[message,above] {
    $x_{0,4}=$\\
    $\delta_{0,4}+\delta_{3,3}$\\
    ${}+\delta_{2,2}+\delta_{1,1}$
  } (w1.west);

\draw[arrow] (w1.south) --
  node[message,right] {
    $x_{1,4}=$\\
    $\delta_{1,4}+\delta_{0,3}$\\
    ${}+\delta_{3,2}+\delta_{2,1}$
  } (w2.north);

\draw[arrow] (w2.west) --
  node[message,below] {
    $x_{2,4}=$\\
    $\delta_{2,4}+\delta_{1,3}$\\
    ${}+\delta_{0,2}+\delta_{3,1}$
  } (w3.east);

\draw[arrow] (w3.north) --
  node[message,left] {
    $x_{3,4}=$\\
    $\delta_{3,4}+\delta_{2,3}$\\
    ${}+\delta_{1,2}+\delta_{0,1}$
  } (w0.south);

\node[update,anchor=south] at (0,4.75)
  {$\sigma_{0,4}=x_{3,4}-\delta_{0,1}$};

\node[update,anchor=south] at (6,4.75)
  {$\sigma_{1,4}=x_{0,4}-\delta_{1,1}$};

\node[update,anchor=north] at (6,-0.75)
  {$\sigma_{2,4}=x_{1,4}-\delta_{2,1}$};

\node[update,anchor=north] at (0,-0.75)
  {$\sigma_{3,4}=x_{2,4}-\delta_{3,1}$};

\end{tikzpicture}%
}
\end{minipage}\hfill
\begin{minipage}[t]{0.49\linewidth}
\centering
\resizebox{\linewidth}{!}{%
\begin{tikzpicture}[
 worker/.style={
   circle, draw=black!75, fill=blue!7,
   line width=0.8pt, minimum size=11mm, font=\large
 },
 arrow/.style={
   -{Stealth[length=2.5mm]},
   line width=0.9pt, blue!65!black
 },
 message/.style={
   font=\large, align=center,
   fill=white, inner sep=3pt
 },
 update/.style={
   font=\large, align=center, text=red!70!black
 }
]
\node[font=\large\bfseries] at (3,6.15) {Round 5};

\node[worker] (w0) at (0,4) {$W_0$};
\node[worker] (w1) at (6,4) {$W_1$};
\node[worker] (w2) at (6,0) {$W_2$};
\node[worker] (w3) at (0,0) {$W_3$};

\draw[arrow] (w0.east) --
  node[message,above] {
    $x_{0,5} = $\\
    $\delta_{0,5}+\delta_{3,4}$\\
    ${}+\delta_{2,3}+\delta_{1,2}$
  } (w1.west);

\draw[arrow] (w1.south) --
  node[message,right] {
    $x_{1,5}=$\\
    $\delta_{1,5}+\delta_{0,4}$\\
    ${}+\delta_{3,3}+\delta_{2,2}$
  } (w2.north);

\draw[arrow] (w2.west) --
  node[message,below] {
    $x_{2,5}=$\\
    $\delta_{2,5}+\delta_{1,4}$\\
    ${}+\delta_{0,3}+\delta_{3,2}$
  } (w3.east);

\draw[arrow] (w3.north) --
  node[message,left] {
    $x_{3,5}=$\\
    $\delta_{3,5}+\delta_{2,4}$\\
    ${}+\delta_{1,3}+\delta_{0,2}$
  } (w0.south);

\node[update,anchor=south] at (0,4.75)
  {$\sigma_{0,5}=x_{3,5}-\delta_{0,2}$};

\node[update,anchor=south] at (6,4.75)
  {$\sigma_{1,5}=x_{0,5}-\delta_{1,2}$};

\node[update,anchor=north] at (6,-0.75)
  {$\sigma_{2,5}=x_{1,5}-\delta_{2,2}$};

\node[update,anchor=north] at (0,-0.75)
  {$\sigma_{3,5}=x_{2,5}-\delta_{3,2}$};

\end{tikzpicture}%
}
\end{minipage}

\caption{ {\small  CERAR protocol with $N=4$, $c=5$, and $L=7$:
rounds 4 and 5. Arrows show transmitted vectors; equations
beside workers show the rolling-state updates after reception
in the indicated round. The update equations appear in \eqref{eq:update_eqns}.}}
\label{fig:CERAR_eg}
\end{figure}
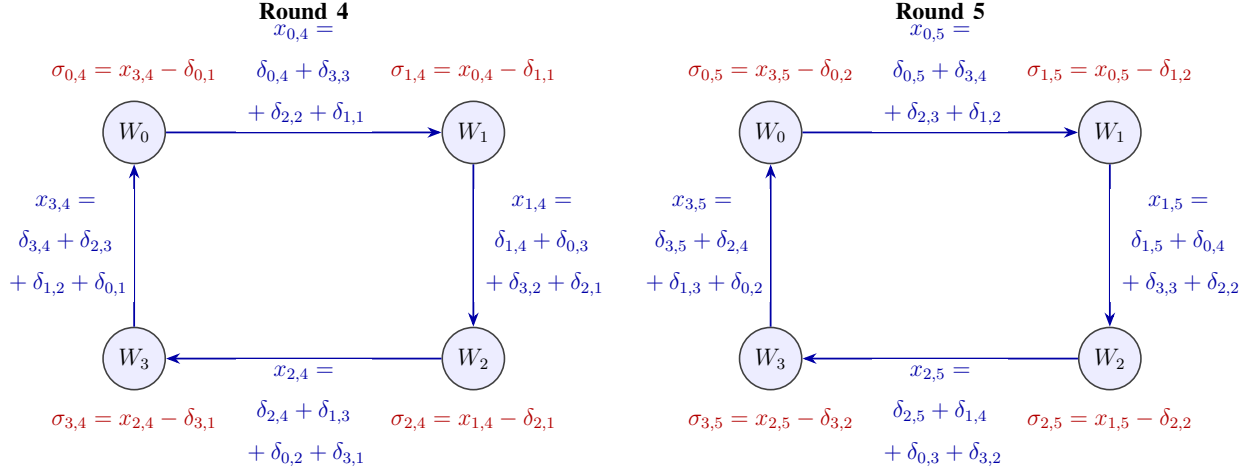
{\bf Problem Formulation:} For a given number of workers $N$, suppose that (for $i=0, \dots, N-1)$ $W_i$ divides its gradient $g_{\calD_i}$ into $c$ chunks. A round consists of clockwise communication over the ring where $W_i$ sends $W_{i+1}$ (indices reduced modulo-$N$) a real-valued vector of length $\tilde{c}$. Suppose that a given protocol communicates over $L$ rounds. The requirement is that worker $W_i$ should be able to reconstruct an approximate gradient $\hat{\bfg}_i$ at the end of these rounds such that
\begin{equation}
    \norm{\hat{\bfg}_i - \nabla L}_2 \leq \tilde{\delta}, \text{~for all $i = 0, \dots, N-1$,} \label{eq:approx_grad}
\end{equation}
where $\tilde{\delta} > 0$ is a small constant. Thus, the total number of real values communicated in the protocol per worker is $L \tilde{c}$. We define the normalized rate of the protocol as $\frac{L\tilde{c}}{d}$, e.g., the basic Ring All-Reduce operates with $L = 2(N-1), c = N, \tilde{c} = d/c$ and with $\tilde{\delta} = 0$ (exact reconstruction).

\noindent {\bf Related Work:}  
Gradient quantization \cite{alistarh2017qsgd,gandikota2022vqsgd,yu2018gradiveq}, and sparsification \cite{alistarh2018convergence, sahu2021rethinking,wangni2018gradient} aim at reducing the communication load of distributed training by ``compressing'' or ``sparsifying'' the gradients and then communicating the compressed representations. These approaches address communication reduction along a direction orthogonal to ours. Specifically, it may be possible to apply our techniques in conjunction with them.


Contributions in the area of gradient coding \cite{tandon_gradient,dimakis_cyclic_mds,YeA18,ramamoorthyMG24,munimR26} aim at making distributed training robust to worker failures/slowdowns; however, these are in the PS + workers model and moreover require redundancy in the assignment of chunks to workers.  Coded distributed computing \cite{LiMA16,kostasR20,huang2026optimalitycodeddistributedcomputing} is similar in the sense that it trades off computational redundancy for reduced communication for exact computation. Our work is related to approximate gradient coding since we rely critically on the judicious usage of linear combinations of gradient chunks and aim at recovering approximate gradients. However, crucially, our setting remains the same as the basic Ring All-Reduce; we do not need redundant chunk assignments.

{\noindent {\bf Main Contributions:}} In this work, we present a \textbf{C}ommunication-\textbf{E}fficient, Approximate \textbf{R}ing \textbf{A}ll-\textbf{R}educe protocol named CERAR. (i) We demonstrate that CERAR can achieve arbitrarily small error for a so-called Vandermonde construction, while operating with normalized communication and storage rates close to 1. The caveat is that this is achieved via increasingly ill-conditioned encoding and decoding matrices. (ii) We characterize the tradeoff between the error and numerical stability of the scheme by explicitly deriving a scheme where the error is $O(\epsilon)$ and the condition number of the matrices is $O(\epsilon^{-r_\star})$ where $r_\star = \lceil c/N\rceil -1$. (iii) We then propose a constrained optimization problem that trades off error for numerical stability in our scheme. Numerical experiments demonstrate the effectiveness of our approach.

\section{Communication Efficient Ring All Reduce (CERAR) Protocol}
\label{sec:prob_form}


The gradients are subdivided into $c$ chunks, i.e., $g_{\calD_i} = [g_{\mathcal D_i}[0]^T ~\cdots ~ g_{\mathcal D_i}[c-1]^T]^T$. In each of the $L = c + N - 2$ rounds (numbered $1$ through $L$), the transmitted vector is of length $p = d/c$ (we assume $c ~|~d$, which can be enforced by zero-padding). 
Let $\bfG_i:=
\begin{bmatrix}
g_{\mathcal D_i}[0]&\cdots&g_{\mathcal D_i}[c-1]
\end{bmatrix}
\in\mathbb R^{p \times c}$ and let $\mathcal{G}^{(i)} = [\bfG_i~|~\dots ~|~\bfG_{i+N-1}]$ (indices reduced modulo-$N$).

{\bf Initialization:} Worker $W_i$ has a $c \times c$ matrix $\bfU_i = \begin{bmatrix}
u_{i,1}&\cdots&u_{i,c}
\end{bmatrix}$. After computing $\bfG_i$, it creates encoded vectors 
\[
\Delta_i:=
\begin{bmatrix}
\delta_{i,1}&\cdots&\delta_{i,c}
\end{bmatrix}
= \bfG_i \bfU_i.
\]
When $\bfU_i$ is nonsingular, the map $\bfG_i\mapsto \Delta_i$ is an
invertible change of basis.  Consequently, the same $d$-length memory region
that initially contains $g_{\calD_i}$ can be transformed in place so
that it contains $\Delta_i$.
We denote the corresponding $p$-length storage slots by
$D_{i,1},\ldots,D_{i,c}$, such that $D_{i,j}$ stores $\delta_{i,j}$. In addition, we allocate $N-1$ additional $p$-length empty dedicated buffers in $W_i$ denoted $B_{i,1}, \dots ,B_{i,N-1}$, and one rolling $p$-length buffer denoted $C_i$. $C_i$
holds the current state used in the next transmission and
can also serve as the scratch buffer for the in-place transform
$\bfG_i\mapsto\Delta_i$ before communication begins. The buffer $C_i$ is initialized to the ``rolling state'' $\sigma_{i,0} = 0$.

The formal definition of the protocol appears in Algorithm \ref{alg:ring-protocol}, which describes the communication and
storage updates at each worker. 
%
For round $k$, where $1 \leq k \leq c$, worker $W_i$ adds the local encoded vector $\delta_{i,k}$ to its rolling state $\sigma_{i,k-1}$
and transmits the result $x_{i,k}$ to $W_{i+1}$, and for $c+1\leq k \leq L$ it simply transmits $x_{i,k} = \sigma_{i,k-1}$. The first $N-1$ received
vectors (in rounds 1 through $N-1$) are stored in dedicated buffers ($B_{i,j}$'s for $W_i$). For round $k=j+N-1$ (where $j \geq 1$),
$W_i$ subtracts its own returning local contribution, namely $\delta_{i,j}$ from the
received vector to obtain the next rolling state $\sigma_{i,k}$, while overwriting
the corresponding local buffer $D_{i,j}$ with the received vector for final
decoding. After $L=c+N-2$ rounds, the stored observations, together
with the final unmodified local encoded vector $\delta_{i,c}$, form $Y_i$, from
which the approximate gradient is reconstructed using $R_i$. The update equations at $W_i$ are (worker indices reduced modulo-$N$)
\begin{equation}
    x_{i,k} = \begin{cases}
    \delta_{i,k} + \sigma_{i,k-1}, & 1 \leq k \leq c,    \\
    \sigma_{i,k-1}, &  c+1 \leq k \leq L
    \end{cases} \text{~{\small \&}~} 
    \sigma_{i,k} = \begin{cases}
        x_{i-1,k}, &  1 \leq k \leq N-1,\\
        x_{i-1,k} - \delta_{i,k-N+1}, &  N \leq k \leq L.
    \end{cases} \label{eq:update_eqns}
\end{equation}

\begin{algorithm}[t]
\caption{Communication and storage updates at worker $W_i$}
\label{alg:ring-protocol}
\begin{algorithmic}[1]
\Require $N \geq 2$, $L=c+N-2$, and encoded vectors
         $\delta_{i,1},\ldots,\delta_{i,c}$
\Ensure Stored observations $Y_i$ for final decoding
\Statex \textit{Worker indices are reduced modulo $N$.
All workers execute each round concurrently.}

\For{$j=1,\ldots,c$}
    \State $D_{i,j} \gets \delta_{i,j}$
\EndFor
\State $C_i \gets 0$ \Comment{Initial state $\sigma_{i,0}$}

\For{$k=1,\ldots,L$}
    \Statex \textit{Step 1: Generate, accumulate, and transmit}
    \If{$k \leq c$}
        \State $C_i \gets C_i+D_{i,k}$
    \EndIf
    \State Transmit $C_i$ to $W_{i+1}$ as $x_{i,k}$
    \Statex \hspace{\algorithmicindent}
        \textit{Complete the outgoing send before overwriting $C_i$.}

    \Statex \textit{Step 2: Receive, store, and update}
    \State Receive $x_{i-1,k}$ from $W_{i-1}$ into $C_i$
    \If{$k \leq N-1$}
        \State $B_{i,k} \gets C_i$
    \Else
        \State $j \gets k-N+1$
        \State $C_i \gets C_i-D_{i,j}$
            \Comment{Cancel the returning local contribution}
        \State $D_{i,j} \gets D_{i,j}+C_i$
            \Comment{Store $x_{i-1,k}$ for decoding}
    \EndIf
    \Statex \hspace{\algorithmicindent}
        \textit{The buffer $C_i$ now contains $\sigma_{i,k}$.}
\EndFor

\State $Y_i \gets
    [\,B_{i,1}\mid\cdots\mid B_{i,N-1}
      \mid D_{i,1}\mid\cdots\mid D_{i,c}\,]$
    \Comment{Logical ordering of the existing buffers}
\end{algorithmic}
\end{algorithm}

\subsection{Illustrative Example with $N=4,c=5, L = 7$}
\label{sec:illustrative_eg}
We explain the basic protocol for the case of four workers and seven rounds, i.e., $N=4, c= 5, L=7$.
In round $1$, worker $W_i$ receives $x_{i-1,1} = \delta_{i-1,1}$ (since $\sigma_{i,0} = 0$, for all $W_i$) from its preceding worker, $W_{i-1}$ and stores it in $B_{i,1}$ . At the end of round $1$, $W_i$ updates its storage to $\sigma_{i,1} = x_{i-1,1} = \delta_{i-1,1}$. 
In round $4$ (see Fig. \ref{fig:CERAR_eg}), $W_i$ receives the following vector from $W_{i-1}$.
\begin{align*}
    x_{i-1,4} &= \delta_{i-1,4} + \sigma_{i-1,3}
    = \delta_{i-1,4} + \delta_{i-2,3} + \delta_{i-3,2} + \delta_{i-4,1}, 
\end{align*}
where the second equality follows by \eqref{eq:update_eqns}.
Note that $i-4 \equiv i \pmod 4$, i.e., $\delta_{i,1}$ needs to be canceled at $W_i$. Therefore, 
\begin{align*}
    \sigma_{i,4} &= x_{i-1,4} - \delta_{i,1} =\delta_{i-1,4} + \delta_{i-2,3} + \delta_{i-3,2}, 
\end{align*}
and $D_{i,1} = x_{i-1,4}$. This process continues until round $5$ (note that $c=5$ here), following which the workers simply forward their stored values. Thus, in round $6$, $W_i$ simply receives $\sigma_{i-1,5}$. For example, $W_0$ receives $\sigma_{3,5} = \delta_{2,5} + \delta_{1,4} + \delta_{0,3}$ and it calculates 
\[
\sigma_{0,6} = \delta_{2,5} + \delta_{1,4}
\]
by canceling $\delta_{0,3}$.
It can be observed that the received vectors at $W_i$ at each round can be compactly represented as matrix $\bfM_i$ of dimension $Nc \times (L+1)$. For example, $\bfM_0$ and $\bfM_1$ are given as

\begin{align}
\mathbf M_0 &=
\begin{bmatrix}
\mathbf 0 &\mathbf 0 & \mathbf 0 &
u_{0,1} & u_{0,2} & u_{0,3} &
u_{0,4} & u_{0,5}\\
\mathbf 0 & \mathbf 0 &
u_{1,1} & u_{1,2} & u_{1,3} &
u_{1,4} & u_{1,5} & \mathbf 0\\
\mathbf 0 &
u_{2,1} & u_{2,2} & u_{2,3} &
u_{2,4} & u_{2,5} &
\mathbf 0 & \mathbf 0\\
u_{3,1} & u_{3,2} & u_{3,3} &
u_{3,4} & u_{3,5} &
\mathbf 0 & \mathbf 0 & \mathbf 0
\end{bmatrix}, \text{and} \label{eq: M_0_matrix}\\
\mathbf M_1 &=
\begin{bmatrix}
\mathbf 0 & \mathbf 0 & \mathbf 0 &
u_{1,1} & u_{1,2} & u_{1,3} &
u_{1,4} & u_{1,5}\\
\mathbf 0 & \mathbf 0 &
u_{2,1} & u_{2,2} & u_{2,3} &
u_{2,4} & u_{2,5} & \mathbf 0\\
\mathbf 0 &
u_{3,1} & u_{3,2} & u_{3,3} &
u_{3,4} & u_{3,5} &
\mathbf 0 & \mathbf 0\\
u_{0,1} & u_{0,2} & u_{0,3} &
u_{0,4} & u_{0,5} &
\mathbf 0 & \mathbf 0 & \mathbf 0
\end{bmatrix}. \label{eq: M_1_matrix}
\end{align}
$\bfM_2$ and $\bfM_3$ can be obtained in a similar manner. 
For instance, the interpretation of the fourth column of $\bfM_0$ is that it receives the $p$-length vector
\[
y_{0,4} = \sum_{\ell=0}^{c-1} u_{0,1}[\ell] g_{\calD_0}[\ell] + \sum_{\ell=0}^{c-1} u_{1,2}[\ell] g_{\calD_1}[\ell] + \sum_{\ell=0}^{c-1} u_{2,3}[\ell] g_{\calD_2}[\ell] + \sum_{\ell=0}^{c-1} u_{3,4}[\ell] g_{\calD_3}[\ell]
\]
from $W_3$ in round 4. $y_{0,4}$ is stored in $D_{0,1}$ after overwriting $\delta_{0,1}$. The last column corresponds to $y_{0,8} = \delta_{0,5}$ that is contained in $D_{0,5}$. It is not overwritten as there are only $L=7$ rounds.

The worker $W_i$ wants to obtain $\hat{\bfg_i}$ (see \eqref{eq:approx_grad}). Thus, $W_0$, e.g., wants to find a decoding vector $\bfr_1$ of length-$(L+1)$ such that it can approximately recover $\sum_{i=0}^{N-1} g_{\calD_i}[0]$, i.e., 
\[
\sum_{i=1}^{L+1} \bfr_{1i} y_{0,i} \approx \sum_{i=0}^{N-1} g_{\calD_i}[0].
\]
This in turn implies that we need $\bfM_0 \bfr_1 \approx [e_0^T~ e_0^T~ e_0^T~e_0^T]^T
= \mathds{1}_4 \otimes e_0$ ($\otimes$ - Kronecker product),
where $e_i, i = 0, \dots, c-1$ denotes the canonical basis vector in $\mathbb{R}^c$. In a similar manner it can be observed that we need the existence of other decoding vectors such that $\mathds{1}_N \otimes e_i$ for $i = 1, \dots, c-1$ can be decoded from each of them. This can be compactly represented as the existence of a $(L+1) \times c$ matrix $\bfR_i$ such that $\bfM_i \bfR_i \approx \mathds{1}_{N} \otimes I_c$ ($\mathds{1}_N$ - all-ones vector of length-$N$, $I_c$ - $c\times c$ identity matrix).
Suppose that the protocol identifies suitable $\bfM_i$ and $\bfR_i$ such that $\norm{\bfM_i \bfR_i - \mathds{1}_{N} \otimes I_c}_F \leq \delta_{\rm dec}$ for $i=0, \dots, N-1$, where $\norm{\cdot}_F$ denotes the Frobenius norm. Then, we have 
\begin{align}
    \norm{\hat{\bfg}_i - \nabla L}_2  = \norm{\mathcal{G}^{(i)} (\bfM_i \bfR_i - \mathds{1}_{N} \otimes I_c)}_F 
&\leq \norm{\mathcal{G}^{(i)}}_2 \delta_{\rm dec} = \norm{\mathcal{G}^{(0)}}_2 \delta_{\rm dec}, \label{eq:bounded_spectral_assumption}
\end{align}
where $\delta_{\rm dec}$ is the decoding error, and $\norm{\cdot}_2$ denotes the spectral norm. The inequality is standard, and the second equality holds since the $\mathcal{G}^{(i)}$'s are related by permutations. 
Thus, assuming that $\norm{\mathcal{G}^{(0)}}_2 \leq C_G$ uniformly over the gradients under consideration, the overall error in the gradient computation ({\it cf.} \eqref{eq:approx_grad}) can be made small enough by choosing $\delta_{\rm dec}$ small enough. Appendix \S{\ref{sec:min_rate_one}} shows that the normalized rate of any linear protocol satisfying \eqref{eq:approx_grad} for sufficiently small $\tilde{\delta}$ (under the bounded spectral norm assumption) is at least 1.

\subsection{Discussion of the CERAR protocol} 




The total storage per worker includes $D_{i,j}, j = 1,\dots,c$, the buffers $B_{i,j}, j = 1, \dots, N-1$ and the buffer $C_i$ that contains the rolling state. 
Thus, the normalized storage rate and communication rate of CERAR are given by
\begin{equation*}
\mathcal{S} 
=1+\frac{N}{c}, \qquad \text{and}\qquad \mathcal{R} = \frac{L}{c}
    =1+\frac{N-2}{c}. 
\end{equation*}
We note that the storage rate of the basic Ring All-Reduce is exactly 1.
Let $\bfU_k, k = 0, \dots, N-1$ denote the collection of the $u_{k,j}, 1 \leq j \leq c$ variables, i.e., $\bfU_k = [u_{k,1}~ u_{k,2}~ \dots ~ u_{k,c}]$ and let $\bfU = (\bfU_0, \dots, \bfU_{N-1})$; these are the encoding matrices. Each $\bfM_i, 0 \leq i \leq N-1$ denotes a $Nc \times (L+1)$ matrix that is obtained as a result of the CERAR protocol (examples in \eqref{eq: M_0_matrix}, \eqref{eq: M_1_matrix}); henceforth we denote it as $\bfM_i(\bfU)$ to show the dependence explicitly. Furthermore, note that the $r$-th block-row of $\bfM_i(\bfU)$ of dimension $c \times (L+1)$ can be expressed as
\begin{align}
    \bfM_i(\bfU)[r,:] &= [\mathbf{0}_{c \times (N-1-r)}~\bfU_{i+r}~ \mathbf{0}_{c \times r}] \text{~for~} r = 0,\dots, N-1.
\end{align}
Here the subscript $i+r$ is reduced modulo-$N$.
Let $\bfR_i, 0 \leq i \leq N-1$ denote $(L+1) \times c$ decoding matrices; let $\bfR = (\bfR_0, \dots, \bfR_{N-1})$, and the target $\bfT = \mathds{1}_{N} \otimes I_c$. For $\bfR_i$, we denote the $c$ consecutive rows starting at row $s$ by $\mathcal{W}_s(\bfR_i) = \bfR_i[s:s+c-1,:]$ for $s=0, \dots, N-1$, where the rows are numbered $0,\dots, L$. The least-squares fit in our scenario is measured by 
\begin{align}
\label{eq:F_unreg}
\Phi(\bfU,\bfR)
=
\sum_{i=0}^{N-1}\norm{\bfM_i(\bfU)\bfR_i-\bfT}_F^2.
\end{align}
We note here that both $\mathcal{R}$ and $\mathcal{S}$ are decreasing in $c$, and the basic Ring All-Reduce already provides a zero-error solution for the case when $c=N$. For the remainder of the paper we only consider the regime $c > N$. Lemma \ref{lemma:nonexistence_exact} is proved in Appendix, \S{ \ref{sec:nonexistence_proof}}. 

\begin{lemma}
\label{lemma:nonexistence_exact}
Suppose that $c>N$.  There is no finite collection of encoding and decoding matrices
$(\bfU, \bfR)$ for which $\Phi(\bfU,\bfR)=0$.
\end{lemma}

\subsubsection{Multiplicative Perturbation Construction (MPC)}
\label{sec:mult_diag_cons}
We now present a specific construction of the $\bfU$ and $\bfR$ matrices that are obtained by (i) a $N \times c$ base matrix $\bfB$ with rows $b_0, \dots, b_{N-1}$, each of length-$c$, and (ii) a $c \times c$ nonsingular matrix $\bfD$ that perturbs $\bfB$ by right-multiplication. Define for $r = 0, \dots, N-1$ and an integer $m$, the row-vector
%
\begin{align}
    a_{r+mN} &= b_r\bfD^m. \label{eq:a_seq}
\end{align}
Thus, the above vector is defined for all integer indices for the LHS of \eqref{eq:a_seq}.
Note that we have $a_{\ell+N} = a_\ell \bfD$. For every integer $\ell$, define the $c\times c$ consecutive-row window
\begin{equation}
A_\ell
=
\begin{bmatrix}
a_\ell\\
a_{\ell+1}\\
\vdots\\
a_{\ell+c-1}
\end{bmatrix}.
\label{eq:A-ell}
\end{equation}
It follows that $A_{\ell+N} = A_\ell \bfD$. We assume that $A_0, \dots, A_{N-1}$ are nonsingular.
Define
\begin{equation}
\bfV_{N-1-k}:=A_k,
\qquad
\bfU_k:=\bfV_k^{-1},
\qquad
k=0,\ldots,N-1.
\label{eq:encoders}
\end{equation}

For worker $W_i$, define the explicit decoder
$\bfR_i\in\mathbb{R}^{(L+1)\times c}$ by taking the $L+1$ consecutive rows: 
\begin{equation}
    \bfR_i[\ell,:]=a_{\ell-i},
    \qquad \ell=0,\ldots,L.
    \label{eq:decoderconstruction}
\end{equation}
Equations \eqref{eq:encoders} and \eqref{eq:decoderconstruction} specify the construction of the scheme; the $s$-th $c$-row window of $\bfR_i$ is
\begin{equation}
\mathcal W_s(\bfR_i)
=
A_{s-i},
\qquad
s=0,\ldots,N-1.
\label{eq:decoderwindow}
\end{equation}


\begin{prop}
\label{prop:block}
Let worker indices be reduced modulo-$N$. For the construction above,
\begin{equation*}
\bfU_{i-1-s}\mathcal W_s(\bfR_i)
=
\begin{cases}
I_c, & s\ge i,\\[1mm]
\bfD^{-1}, & s<i.
\end{cases}
\end{equation*}
\end{prop}
\begin{proof}
If \(s\ge i\), then \(s-i\in\{0,\ldots,N-1\}\).  By  \eqref{eq:encoders},
\begin{align*}
A_{s-i} &=\bfV_{i-1-s}, \text{(reducing the index of RHS modulo-$N$)}\\
\text{Therefore,~} \bfU_{i-1-s}\mathcal W_s(\bfR_i)
&=
\bfV_{i-1-s}^{-1}A_{s-i} = I_c.
\end{align*}
If \(s<i\), we can assert $A_{s-i}
= A_{s-i+N}\bfD^{-1}$ where $s-i+N\geq 0$ since $s-i \geq -(N-1)$.
Again,
\begin{align*}
A_{s-i+N}&=\bfV_{i-1-s}, \text{~and hence}\\
\bfU_{i-1-s}\mathcal W_s(\bfR_i)
&= \bfV_{i-1-s}^{-1}A_{s-i+N}\bfD^{-1} = \bfD^{-1}.  \qquad \qquad\qquad\qquad \qedhere
\end{align*}
\end{proof}

\begin{prop}
\label{prop:explicitresidual}
Let $\tilde{\bfU}$ and $\tilde{\bfR}$ denote the collections of matrices constructed as in \eqref{eq:encoders} and \eqref{eq:decoderconstruction}. Then,
\begin{equation}
\inf_{\bfU, \bfR} \Phi(\bfU,\bfR) \leq \Phi(\tilde{\bfU},\tilde{\bfR})
=
\frac{N(N-1)}{2}
\|\bfD^{-1}-I_c\|_F^2.
\label{eq:explicitresidual}
\end{equation}
\end{prop}

\begin{proof}
Worker \(W_i\) has exactly \(i\) indices \(s\in\{0,\ldots,N-1\}\) for which
\(s<i\).  By Proposition~\ref{prop:block}, these are precisely the
nonzero residual blocks, and each is equal to
$\bfD^{-1}-I_c$.  Therefore
\[
\Phi(\tilde{\bfU},\tilde{\bfR})
=
\sum_{i=0}^{N-1}
i\,\|\bfD^{-1}-I_c\|_F^2
=
\frac{N(N-1)}2
\|\bfD^{-1}-I_c\|_F^2.
\]
The matrices $\tilde{\bfU}$ and $\tilde{\bfR}$ specified in \eqref{eq:encoders} and \eqref{eq:decoderconstruction} are feasible choices in the
least-squares minimization in \eqref{eq:F_unreg}, so the optimized
least-squares residual cannot exceed their residual.
\end{proof}



\subsubsection{A Vandermonde construction that achieves arbitrarily small $\Phi(\bfU, \bfR)$}
\label{sec:vandermonde}
Let $\xi_1,\ldots,\xi_c$ be distinct reals, and for $\tau>0$, set $z_j(\tau):=e^{\tau\xi_j}$ for $j = 1, \dots, c$. Define
%
\begin{align*}
    b_\ell(\tau)
    &:=
    \begin{bmatrix}
        z_1(\tau)^\ell & \cdots & z_c(\tau)^\ell
    \end{bmatrix}, \text{~for $\ell = 0, \dots, N-1$, and~} \\
    \bfD_\tau
    &:=\diag\bigl(z_1(\tau)^{N},\ldots,z_c(\tau)^{N}\bigr)
    =\diag\bigl(e^{N\tau\xi_1},\ldots,e^{N\tau\xi_c}\bigr). 
\end{align*}
This defines the sequence specified in \eqref{eq:a_seq}, and therefore $A_\ell$. Henceforth, we refer to it as $A_\ell(\tau)$ to emphasize the dependence on $\tau$; it is a shifted Vandermonde matrix with nodes $z_1(\tau), \dots z_c(\tau)$ and hence nonsingular \cite{horn_matrix_analysis}. In particular, this choice satisfies the requirement that $A_0(\tau), \dots, A_{N-1}(\tau)$  are nonsingular. We call this the Vandermonde construction. The following theorem follows from an application of Proposition \ref{prop:explicitresidual} (details in Appendix, \S{\ref{sec:vand_cons_proof}}).

\begin{theorem}
    Let $\tilde{\bfU}_\tau$ and $\tilde{\bfR}_\tau$ be the matrices obtained by applying the Vandermonde construction. Then, $\Phi(\tilde{\bfU}_\tau,\tilde{\bfR}_\tau) =  \frac{N^3(N-1)}{2}\tau^2\sum_{j=1}^{c}\xi_j^2 + O(\tau^3) \to 0$ as $\tau \downarrow 0$ so that  $\inf_{\bfU, \bfR} \Phi(\bfU,\bfR)=0$. If $c>N$, the infimum is not attained.
\end{theorem}



\begin{remark}

Choosing a sequence $\tau_c$ (dependent on $c$) such that $\tau_c^2\sum_{j=1}^{c}\xi_j^2 \to 0$, this shows a rather striking fact that one can simultaneously achieve a normalized rate and storage of $1$ as $c \to \infty$\footnote{We need $d \to \infty$ as well. In practice $d$ is many orders of magnitude larger than $c$, so this is not a limitation.} while driving the approximation error in the reconstruction arbitrarily close to zero. To our best knowledge, this is the first work that shows such a result for the well-known All-Reduce setting.
\end{remark}

\begin{remark} 
As $\tau\downarrow 0$, all Vandermonde nodes satisfy $z_j(\tau)\longrightarrow 1$, so that $\bfV_k(\tau) \bfV_k(\tau)^T \to c \mathds{1}\mathds{1}^T$. Thus, the condition number (ratio of maximum to minimum singular values) $\kappa_2(\bfV_k(\tau))
    =\kappa_2(\bfU_k(\tau)) \longrightarrow\infty$, i.e.,  this construction is numerically unstable for small $\tau$. Thus, in practice, one needs to adapt the approach appropriately.

\end{remark}

\subsection{Trade-off between LS Error and conditioning}
We now examine the role of the multiplicative perturbation $\bfD$ on the conditioning of the $\bfM_i$ matrices. Let $\bfD_\epsilon = I_c + \epsilon \bfH$ with $\epsilon \norm{\bfH}_2 < 1$ for a fixed matrix $\bfH$. This implies that $\norm{\bfD_\epsilon^{-1}}_2 = \frac{1}{\sigma_{\min}(\bfD_\epsilon)} \leq \frac{1}{1-\epsilon\norm{\bfH}_2}$. Moreover, 
$\bfD_\epsilon^{-1} - I_c = - \epsilon\bfD_\epsilon^{-1} \bfH$, so that $\norm{\bfD_\epsilon^{-1} - I_c}_F \leq \epsilon\norm{\bfD_\epsilon^{-1}}_2 \norm{\bfH}_F$.

Applying the result of Proposition \ref{prop:explicitresidual} (nonsingularity of the corresponding $A_\ell(\epsilon)$ follows from the forthcoming assumptions and discussion), we can conclude 
\begin{equation}
\Phi(\bfU,\bfR) \leq \frac{N(N-1)}{2} \frac{\epsilon^2\norm{\bfH}_F^2}{(1 - \epsilon\norm{\bfH}_2)^2} = O(\epsilon^2).    \label{eq:error_scaling}
\end{equation}

Consider the $c \times c$ consecutive-row window $A_\ell(\epsilon)$ where $\ell \in \{0, \dots, N-1\}$ that is formed by using $\bfB$ and $\bfD_\epsilon$ ({\it cf.} \eqref{eq:A-ell}). These are the kind of window types that arise in our construction of the decoding matrices $\bfR_i$'s, except that they are sometimes right-multiplied by powers of $\bfD_\epsilon$. However, $\kappa_2(\bfD_\epsilon) = 1 + O(\epsilon)$ (condition number - ratio of maximum to minimum singular value) so it does not affect the scaling of $\kappa_2(A_\ell(\epsilon)).$

Let $c = \beta N + s$, where $0 \leq s < N$ and define $r_\star = \left\lceil \frac{c}{N} \right\rceil - 1$. Let $E_\ell = \{\ell, \ell+1, \dots, \ell+s-1\} \pmod N$. For $j \in \{\ell, \ell+1, \dots, \ell+c-1\}$, consider its remainder upon division by $N$. If $s> 0$, then a given remainder $\alpha \in \{0, \dots, N-1\}$ occurs $n_\alpha(\ell)$ times where
\begin{align*}
    n_\alpha(\ell) &= \begin{cases}
        \beta + 1, & \text{~if~}\alpha \in E_\ell\\
        \beta, & \text{~if~} \alpha \notin E_\ell.
    \end{cases}
\end{align*}
If $s=0$, then every remainder occurs exactly $\beta$ times. For $\ell\in\{0,\ldots,N-1\}$, the smallest index
$j\geq\ell$ satisfying $j\equiv\alpha\pmod N$ is
$j=\alpha+Nm_\alpha(\ell)$, where
\[
    m_\alpha(\ell)=
    \begin{cases}
        0, & \alpha\geq\ell,\\
        1, & \alpha<\ell.
    \end{cases}
\]
Thus, the first row associated with remainder $\alpha$
is $b_\alpha\bfD_\epsilon^{m_\alpha(\ell)}$, followed by $b_\alpha \bfD_{\epsilon}^{m_\alpha(\ell)+1}, \dots, b_\alpha \bfD_{\epsilon}^{m_\alpha(\ell) + n_\alpha(\ell) -1}$.


Let $F_n$ denote the $n \times n$ lower-triangular finite-difference matrix (Appendix, \S{\ref{sec:finite_diff_fact}}), and $\Pi_\ell$ be a permutation matrix that groups the rows of $A_\ell$ by remainder class, and define $\tilde{F}_\ell = \text{diag}(F_{n_0(\ell)}, F_{n_1(\ell)}, \dots, F_{n_{N-1}(\ell)})$. Then, using Lemma \ref{lem:difference} (Appendix, \S{\ref{sec:finite_diff_fact}}) we have
\begin{align}
    \tilde{F}_\ell \Pi_\ell A_\ell(\epsilon) &= \Lambda_\ell(\epsilon) K_\ell(\epsilon) \label{eq:diagonalize_arg}
\end{align}
where $\Lambda_\ell(\epsilon)$ is a diagonal matrix with entries $\{\epsilon^j: j=0, \dots, n_r(\ell) - 1, r= 0, \dots, N-1\}$ and the corresponding rows of $K_\ell(\epsilon)$ are $b_r \bfD_\epsilon^{m_r(\ell)} \bfH^j$. We note that as $\epsilon \downarrow 0$, we have $\bfD_\epsilon^{m_r(\ell)} \to I_c$. This implies that up to a row-permutation, $K_\ell(\epsilon)$ approaches a matrix denoted $K^{(\ell)}(\bfB, \bfH)$ where

\begin{equation*}
K^{(\ell)}(\bfB, \bfH)=\begin{bmatrix}
\bfB\\
\bfB \bfH\\
\vdots\\
\bfB \bfH^{\beta-1}
\end{bmatrix} \text{~if~$s=0$,} \quad \text{and} \quad 
K^{(\ell)}(\bfB, \bfH) = 
\begin{bmatrix}
\bfB\\
\bfB \bfH\\
\vdots\\
\bfB \bfH^{\beta-1}\\
(\bfB \bfH^\beta)_{E_\ell}
\end{bmatrix} \text{~if~ $s > 0$.}
\end{equation*}
Here \((\bfB \bfH^\beta)_{E_\ell}\) denotes the \(s\) rows indexed by the set \(E_\ell\). An example showing the calculation in \eqref{eq:diagonalize_arg} appears in Appendix \S{\ref{sec:finite_diff_fact}}.

It is not too hard to see that the maximum singular value of $K^{(\ell)}(\bfB,\bfH)$ is upper bounded by $C_1 \sigma_{\max}(\bfB)$ where $C_1$ is a constant that depends on $r_\star$ and $\sigma_{\max}(\bfH)$ (see Appendix, \S{\ref{sec:sing_val_K}}).

{\bf Assumption:} We assume that $K^{(\ell)}(\bfB, \bfH)$ is nonsingular and further that $\min_{\ell =0, \dots, N-1} \sigma_{\min}(K^{(\ell)}(\bfB, \bfH)) = \gamma > 0$ . This can be justified through a Schwartz-Zippel type argument, and through multiple trials for selecting $\bfH$. Constructions of $\bfH$ that depend on $\bfB$ can also be considered.




In the following theorem, $N,c,\bfB$, and $\bfH$ are fixed, and all asymptotic notation refers to $\epsilon \to 0$.
The proof appears in the Appendix, \S{\ref{sec:proof_singular_hierarchy}}.
\begin{theorem}[Singular value scaling]
\label{thm:singularhierarchy}
The singular values of $A_\ell(\epsilon)$ have the same powers of
$\epsilon$, up to constant factors, as the diagonal entries of
$\Lambda_\ell(\epsilon)$. Specifically, (i) If $s=0$, then for each $j=0,\ldots,\beta-1$, exactly $N$
singular values are $\Theta(\epsilon^j)$. (ii) If $s>0$, then for each $j=0,\ldots,\beta-1$, exactly $N$
singular values are $\Theta(\epsilon^j)$ and the remaining $s$ singular values are $\Theta(\epsilon^\beta).$

\end{theorem}

\begin{cor}[Exact smallest-singular-value exponent]
\label{cor:smin}
For $r_\star=\left\lceil\frac{c}{N}\right\rceil-1$ and every
\(\ell=0,\ldots,N-1\),
\begin{equation*}
\sigma_{\min}(A_\ell(\epsilon))
=
\Theta(\epsilon^{r_\star}),
\qquad
c_0 \leq \sigma_{\max}(A_\ell(\epsilon))
\leq \tilde{C}_1 \norm{B}_2,
\end{equation*}
where $c_0$ is a constant independent of $\epsilon$. Thus, $\sigma_{\max}(A_\ell(\epsilon)) = \Theta(1)$.
Thus, $\kappa_2(A_\ell(\epsilon)) = \Theta(\epsilon^{-r_\star})$. Furthermore, since the $\bfU_i$'s are formed as inverses of matrices of the form $A_\ell(\epsilon)$, we also have that  $\kappa_2(\bfU_i) = \Theta(\epsilon^{-r_\star})$. From the discussion in Appendix, \S{\ref{sec:reln_U_M_cond_no}} we can conclude that $\kappa_2(\bfM_i) =  O(\sqrt{N} \epsilon^{-r_\star})$ for all $i = 0, \dots, N-1$.
\end{cor}
\begin{remark}
    Corollary \ref{cor:smin} does not apply to the construction in Section \ref{sec:vandermonde} since the Vandermonde construction does not work with a fixed $\bfB$ and $\bfH$. In particular, the corresponding definitions of $\bfB$ and $\bfH$ change with the parameter $\tau$ that goes to zero.
\end{remark}

\subsection{Constrained Optimization Problem Formulation}
\label{sec:constrained_opt}
The construction in Section \ref{sec:mult_diag_cons} guarantees LS squared error $O(\epsilon^2)$, and $O(\sqrt{N} \epsilon^{-r_\star})$ conditioning of the $\bfM_i$ matrices. 
Using this solution as an initialization for the following constrained optimization, it is possible to trade off these competing objectives. Consider
%
\begin{align}
\label{eq:F}
\text{minimize~} & \Phi(\bfU,\bfR) \nonumber \\
\text{subject to~} &\kappa_2(\bfM_i(\bfU)) \leq C, \text{~for~} i=0, \dots, N-1.
\end{align}
Although \eqref{eq:F} is nonconvex, we use an alternating least-squares heuristic in which the $\bfU_i$'s are projected onto a stricter spectral constraint set that is sufficient to guarantee the condition-number constraint.
%
In particular, in one iteration, we minimize $\Phi(\bfU,\bfR)$ with respect to $\bfU$ while keeping $\bfR$ fixed, and vice versa. 
After this, we project the computed singular values of the $\bfU_i$'s into a common interval $a_0 \leq \sigma_{\min}(\bfU_i) \leq \sigma_{\max}(\bfU_i) \leq b_0$.
Note that this implies that $\kappa_2(\bfM_i) \leq \sqrt{N}b_0/a_0$. Therefore, the constant $C$ in \eqref{eq:F} needs to be chosen so that $\sqrt{N}b_0/a_0 \leq C$.

\begin{table}[t]
    \centering
    \caption{\small Raw reduction time in milliseconds; lower is
    better. M $:=10^6$, B $:=10^9$.}
    \label{tab:reduction_time}
    \small
    \begin{tabular}{llrrr}
        \toprule
        System / interconnect & Tensor size & NCCL & Our algorithm & Gain \\
        \midrule
        A100 PCIe, 8 GPUs / PCIe & 100M
        & 1,775.9 & 940.1 & $47.1\%$ \\
        A100 PCIe, 8 GPUs / PCIe & 500M
        & 8,611.2 & 4,651.0 & $46.0\%$ \\
        A100 PCIe, 8 GPUs / PCIe & 1B
        & 17,954.4 & 9,558.6 & $46.8\%$ \\
        \midrule
        H200 SXM, 4 GPUs / NVLink & 9B
        & 147.250 & 197.097 & $-33.9\%$ \\
        H200 SXM, 4 GPUs / NVLink & 15B
        & 244.974 & 327.733 & $-33.8\%$ \\
        H200 SXM, 4 GPUs / NVLink & 24B
        & 392.049 & 522.366 & $-33.2\%$ \\
        \bottomrule
    \end{tabular}
\vspace{-5mm}
\end{table}

\section{Numerical Experiments}
\label{sec:expts_num}

Our numerical experiments compared the performance of the CERAR protocol to the production-level NCCL All-Reduce by \cite{nvidia_nccl}. We considered GPU clusters under two different settings -  a low-bandwidth  (PCIe, 31.5 GB/s) and a high-bandwidth (NVLink, $\geq 300$ GB/s) interconnect. 
All code for recreating these experiments can be found in an anonymous repository. \footnote{\url{https://anonymous.4open.science/r/CERAR--43C6/}} The CERAR scheme was fit using the methods in Section \ref{sec:constrained_opt} (details in Appendix, \S{\ref{sec:appendix_num_exp}}).


{\bf All-Reduce Phase time:} 
We ran the evaluations on single-node systems with eight A100 GPUs
and four H200 GPUs; each GPU started with a FP32 local tensor. CERAR used $c=20,L=26$ for the A100 experiments and $c=12,L=14$
for the H200 experiments. We measured the elapsed time for each GPU to complete
the reduction, including communication and local computation.
For our algorithm, this includes (i) encoding (preparing $\delta_{i,j}$'s), (ii) all communication rounds and their local updates, and (iii)
decoding using the stored values and the $\mathbf{R}_i$ matrices to reconstruct
the approximate sum.
As shown in Table \ref{tab:reduction_time}, CERAR outperforms the NCCL implementation on the
A100 PCIe system, reducing the reduction time
by 46.0--47.1\%. However, on the H200 NVLink system, CERAR is 33.2\%--33.9\% slower than NCCL. Appendix \S{\ref{sec:appendix_num_exp}} has more details and more fine-grained information about the time taken for the different parts of the All-Reduce phase. 

We emphasize that NCCL is a highly optimized (open-source) NVIDIA implementation. In particular, it establishes multiple communication channels and incorporates overlapping computation and communication; these features greatly improve its performance. Our implementation of CERAR is meant more as a proof of concept and the implementation has not been optimized anywhere close to the level of NCCL. In addition, we emphasize that the gap between the methods will typically manifest in the regime when the parameter length is very high. Owing to access limitations, we were unable to benchmark reduction tensors $>$ 24 billion FP32 entries. 
However, we believe that an optimized implementation can yield tangible benefits even in the high-bandwidth NVLink systems.

{\bf Training time:} We trained a Pythia-1.4B model \cite{pmlr-v202-biderman23a}, containing
$\approx 1.41$ billion parameters, on the TinyStories dataset \cite{eldan2023tinystoriessmalllanguagemodels}. Both A100 systems used four GPUs with CERAR using $c=12, L = 14$. For the PCIe experiment, we trained for 1800 iterations ($\approx 117.96 \times 10^6$ tokens), and for the SXM experiment, we trained for 7630 iterations ($\approx 1\times 10^9$ tokens); more details in Appendix \S{\ref{sec:appendix_num_exp}}.

As shown in Table~\ref{tab:training_time}, on the PCIe configuration,
NCCL requires 3.79 hours, whereas CERAR
requires 3.14 hours; the corresponding validation losses at the end of the 1800 iterations are 1.4218 and 1.4213 respectively. Thus, our algorithm reduces the overall training
time by 17.3\%, saving approximately 39 minutes.
On the SXM configuration, NCCL and our algorithm require 4.578 and
4.745 hours, respectively. The validation losses of both schemes are approximately 1.16449. However, our algorithm takes 3.65\% longer.

{\bf Gradient Error:} In the SXM experiment, we measured the relative
$\ell_2$ error of the reduced gradient at every training step.
Table~\ref{tab:gradient_error} reports the
maximum error across workers at the first, middle, and final
steps. For our algorithm, these errors are approximately
1.61\%, 2.87\%, and 2.72\%, respectively, while those for NCCL
are below $7\times10^{-9}$ (consistent with floating-point roundoff). Nevertheless, these errors do not impact the training in an appreciable manner.


\begin{table}[t]
    \centering
    \caption{{\small Training results in $W_0$ for Pythia-1.4B on two four-GPU
    configurations.}}
    \label{tab:training_time}
    \small
    \begin{tabular}{lrrrr}
        \toprule
        Method &
        Training time &
        Throughput &
        Final val.\ loss &
        Peak memory \\
        &
        (hours) &
        (tokens/s) &
        &
        (GB) \\
        \midrule
        \multicolumn{5}{c}{$4\times$ A100 PCIe 40GB
        (PCIe interconnect): 1,800 steps} \\
        \midrule
        NCCL All-Reduce
        & 3.79 & 8,644 & 1.4218 & 36.2 \\
        Our algorithm
        & 3.14 & 10,451 & 1.4213 & 38.6 \\
        \midrule
        \multicolumn{5}{c}{$4\times$ A100 SXM 80GB
        (NVLink interconnect): 7,630 steps} \\
        \midrule
        NCCL All-Reduce
        & 4.578 & 60,682 & 1.1644904 & 49.55 \\
        Our algorithm
        & 4.745 & 58,543 & 1.1644902 & 51.91 \\
        \bottomrule
    \end{tabular}
\vspace{-2mm}    
\end{table}

\begin{table}[t]
    \centering
    \caption{ {\small Relative $\ell_2$ gradient error on four A100 SXM
    80GB GPUs, measured against a separate NCCL reference
    reduction and maximized across workers.}}
    \label{tab:gradient_error}
    \small
    \begin{tabular}{lccc}
        \toprule
        Method & Step 1 & Step 3,815 & Step 7,630 \\
        \midrule
        NCCL All-Reduce
        & $6.772\times10^{-9}$
        & $5.532\times10^{-9}$
        & $6.437\times10^{-9}$ \\
        Our algorithm
        & $0.016078$
        & $0.028653$
        & $0.027179$ \\
        \bottomrule
    \end{tabular}
\vspace{-5mm}    
\end{table}
\section{Conclusions and Future Work}

This work demonstrates that Ring All-Reduce can be made much more communication-efficient when approximate recovery of the sum is sufficient and presents theoretical results guaranteeing asymptotic optimality amongst the class of linear protocols, and practical constructions that guarantee performance metrics (error and numerical stability). Future work will involve optimized implementations of our protocol that allow for practical benefits even in high-bandwidth GPU clusters.

\bibliographystyle{IEEEtran}

\newpage
\section{appendix}

\subsection{Proof that any linear protocol needs to have a normalized rate of at least 1}
\label{sec:min_rate_one}
We say that a protocol is linear if it operates by linear encoding and decoding at the workers.
Here, we prove that any linear protocol that satisfies the requirement in \eqref{eq:approx_grad} has to have rate at least 1. We proceed by contradiction; assume that there exists a linear protocol satisfying \eqref{eq:approx_grad} for all admissible gradients having $L\tilde{c} < d$. In what follows, we will construct two different gradient vectors such that the input to $W_i$ is the same, but the actual values of the sum of the gradients are very different.

We note that at the end of the rounds worker $W_i$ reconstructs an approximation to $\nabla L$ from the incoming vectors obtained from $W_{i-1}$. Suppose that we stack all the vectors received by $W_i$ in a long vector of length $L \tilde{c}$, denoted $\beta_i$ and suppose that $L\tilde{c} < d$. Now consider the linear map from the gradients outside $W_i$ (namely $g_{\calD_j}, j \neq i$) to $\beta_i$, and denote it $\Gamma_i$. This map $\Gamma_i$ is from $\mathbb{R}^{(N-1)d} \to \mathbb{R}^{L \tilde{c}}$. 

Now, we restrict our attention to gradients where every other worker has the same gradient $v \in \mathbb{R}^d$ and $g_{\calD_i}=0$. This provides a restricted linear map from $\mathbb{R}^d$ to $\mathbb{R}^{L \tilde{c}}$ denoted $\tilde{\Gamma}_i$.

Since $L \tilde{c} < d$, there exists a nonzero vector $h \in \mathbb{R}^{d}$ such that $\tilde{\Gamma}_i (h) = 0$. Let $t > 0$. Now, consider two different sets of gradient vectors: (i) $g_{\calD_j} = t h$ for $j\neq i$, and (ii) $g_{\calD_j} = -t h$ for $j\neq i$. Since $\tilde{\Gamma}_i(h) = 0$, worker $W_i$ sees the same set of input vectors in both cases. Set 
\begin{align*}
t & = \frac{C_G}{\sqrt{N-1} \norm{h}_2}.     
\end{align*}
For both sets of inputs $\norm{\mathcal{G}^{(0)}}_2  \leq \norm{\mathcal{G}^{(0)}}_F = \sqrt{N-1} t \norm{h}_2 = C_G$. This means that both sets of gradient vectors satisfy the bounded spectral norm assumption in \eqref{eq:bounded_spectral_assumption}. However, their corresponding sums differ in norm by $2 C_G \sqrt{N-1}$. Thus, $W_i$ cannot recover the gradient within $\ell_2$-error norm $\tilde{\delta}$ (see \eqref{eq:approx_grad}) that is small enough for both sets of gradient vectors. This gives the required contradiction.

\subsection{Proof of Lemma \ref{lemma:nonexistence_exact}}
\label{sec:nonexistence_proof}
Assume, for contradiction, that $\Phi(\bfU,\bfR)=0$. Then, we have
\[
    \bfM_i(\bfU)\bfR_i=\bfT,
\]
for all $W_i$. We note that the $r$-th block of $\bfM_i(\bfU) \bfR_i$ is given by $\bfU_{i+r} \mathcal{W}_{N-1-r}(\bfR_i)$ for $r =0, \dots, N-1$ (worker indices reduced modulo-$N$). This can also be expressed as $\bfU_{i-1-s} \mathcal{W}_{s}(\bfR_i)$ by setting $s=N-1-r$. This requires
\begin{equation}
    \bfU_{i-1-s} \mathcal{W}_{s}(\bfR_i)=I_c,
    \qquad s=0,\ldots,N-1.
    \label{eq:exactblocks}
\end{equation}
Consequently, all $\bfU_i, i =0,\ldots, N-1$ must be nonsingular.  Defining $\bfV_i:=\bfU_i^{-1}$, eq. \eqref{eq:exactblocks} becomes
\begin{equation}
    \mathcal{W}_{s}(\bfR_i)=\bfV_{i-1-s}.
    \label{eq:exactwindows}
\end{equation}
Consecutive windows of the matrices $\bfR_i$ for $i = 0, \dots, N-1$ overlap in $c-1$ rows.  Taking
consecutive values of $s$ in \eqref{eq:exactwindows} for all workers, therefore gives, for
every $k$,
\begin{equation*}
    \bfV_k[2:c,:]=\bfV_{k-1}[1:c-1,:],
\end{equation*}
where the row indices of $\bfV_k$ are one-based.

Equivalently, row by row,
\[
    \bfV_k[t+1,:]=\bfV_{k-1}[t,:],
    \qquad t=1,\ldots,c-1.
\]

For any $j=1,\ldots,c-N$, applying this identity $N$ times yields
\begin{align*}
    \bfV_k[j+N,:]
    &=\bfV_{k-1}[j+N-1,:] \notag\\
    &=\cdots
      =\bfV_{k-N}[j,:]
      =\bfV_k[j,:],
\end{align*}
where the last equality holds because $k \equiv (k-N) \pmod N$.
Consequently, every row after the first $N$ repeats an earlier
row with the same index modulo $N$. Hence
\[
    \operatorname{rank}(\bfV_k)\leq N<c,
\]
contradicting the nonsingularity of $\bfV_k$.
%

\subsection{Error bound of Vandermonde Construction}
\label{sec:vand_cons_proof}
Let $\tilde{\bfU}_\tau$ and $\tilde{\bfR}_\tau$ be the matrices obtained by applying the Vandermonde construction. Applying the result of Proposition \ref{prop:explicitresidual}, we have 
\begin{align}
    0\leq\Phi(\tilde{\bfU}_\tau,\tilde{\bfR}_\tau)
    &=
    \frac{N(N-1)}{2}
    \sum_{j=1}^{c}
       \left(e^{-N\tau\xi_j}-1\right)^2 \nonumber\\ 
       &=     \frac{N^3(N-1)}{2}\,
    \tau^2\sum_{j=1}^{c}\xi_j^2
    +O(\tau^3) \text{~(when $\tau$ is small).}
    \label{eq:upperbound}
\end{align}
For fixed $\xi_j$, the right-hand side converges to zero quadratically as
$\tau \downarrow 0$.  Hence $\inf_{\bfU,\bfR} \Phi(\bfU,\bfR)=0$. When $c>N$, Lemma \ref{lemma:nonexistence_exact} shows that no finite $(\bfU,\bfR)$ attains this value.

\subsection{Proof of Theorem \ref{thm:singularhierarchy}}
\label{sec:proof_singular_hierarchy}

We first note that since $K_\ell(\epsilon) \to P_\ell K^{(\ell)}(\bfB, \bfH)$ for an appropriate row-permutation $P_\ell$, we have by the continuity of singular values and the assumption that $\min_\ell \sigma_{\min}(K^{(\ell)}(\bfB, \bfH)) = \gamma > 0$, there exist $\epsilon$ small enough so that $\sigma_{\min}(K_\ell(\epsilon)) \geq \gamma/2$ for every $\ell$.

From \eqref{eq:diagonalize_arg},
\[
A_\ell(\epsilon)
=
\Pi_\ell^T \tilde{F}_\ell^{-1}
\Lambda_\ell K_\ell(\epsilon).
\]
The matrix \(\Pi_\ell\) is orthogonal, and \(\tilde{F}_\ell\) is fixed and invertible. For every singular-value index \(i\),
standard product inequalities \cite{horn_matrix_analysis} give
\[
\sigma_{\min}(K_\ell(\epsilon))
\sigma_i(\Lambda_\ell(\epsilon))
\le
\sigma_i(\Lambda_\ell(\epsilon) K_\ell(\epsilon))
\le
\sigma_{\max}(K_\ell(\epsilon))\,
\sigma_i(\Lambda_\ell(\epsilon)).
\]
Similarly, multiplication on the left by the fixed matrix
\(\tilde{F}_\ell^{-1}\) changes every singular value by at most fixed positive
multiplicative constants:
\[
\sigma_{\min}(\tilde{F}_\ell^{-1})\,
\sigma_i(\Lambda_\ell (\epsilon)) \sigma_{\min} (K_\ell(\epsilon))
\le
\sigma_i(\tilde{F}_\ell^{-1}\Lambda_\ell K_\ell(\epsilon))
\le
\sigma_{\max}(\tilde{F}_\ell^{-1})\,
\sigma_i(\Lambda_\ell K_\ell).
\]
Every singular value of \(A_\ell\) is therefore
bounded above and below by positive constants multiplied by the corresponding
singular value of \(\Lambda_\ell\).  The multiplicities follow directly from
the characterization of $\Lambda_\ell$.

\subsection{Finite-difference factorization}
\label{sec:finite_diff_fact}
For \(n\ge1\), define the \(n\times n\) lower-triangular finite-difference
matrix \(F_n\)
\begin{equation*}
(F_n)_{j+1,k+1}
=
\begin{cases}
(-1)^{j-k}\binom{j}{k}, & 0\le k\le j,\\
0, & k>j,
\end{cases}
\end{equation*}
where $j,k \in \{0, \dots, n-1\}$ and the rows and columns of $F_n$ are numbered $1, \dots, n$.
The matrix \(F_n\) is invertible and
\[
\det F_n=1.
\]


\begin{lemma}[Exact finite-difference identity]
\label{lem:difference}
For every row vector \(b\), every integer \(m\ge0\), and every \(j\ge0\),
\begin{equation*}
\sum_{k=0}^j
(-1)^{j-k}\binom{j}{k}
b\bfD_\epsilon^{m+k}
=
\epsilon^j
b\bfD_\epsilon^m H^j.
\end{equation*}
\end{lemma}

\begin{proof}
Factor \(b\bfD_\epsilon^m\) on the left and use the binomial identity:
\[
\sum_{k=0}^j
(-1)^{j-k}\binom{j}{k}
\bfD_\epsilon^k
=
(\bfD_\epsilon-I_c)^j
=
\epsilon^jH^j.
\]
\end{proof}
We show an example of how the finite difference matrix acts below.
\begin{example}
\label{eg:finite_diff_n_3}
    For $n=3$, we have
    \[
    F_3 = \begin{bmatrix}
        1 & 0 & 0\\
        -1 & 1 & 0\\
        1 & -2 & 1\\
    \end{bmatrix}.
    \]
    Thus, we have
    \begin{align*}
        F_3 \begin{bmatrix}
            a_0\\
            a_0 \bfD_\epsilon\\
            a_0 \bfD_\epsilon^2            
        \end{bmatrix} &= \begin{bmatrix}
        1 & 0 & 0\\
        -1 & 1 & 0\\
        1 & -2 & 1\\
    \end{bmatrix} \begin{bmatrix}
            a_0\\
            a_0 \bfD_\epsilon\\
            a_0 \bfD_\epsilon^2            
        \end{bmatrix} \\
        &= \begin{bmatrix}
            a_0\\
            \epsilon a_0 \bfH\\
            \epsilon^2 a_0 \bfH^2\\
        \end{bmatrix}\\
        &= \begin{bmatrix}
            1 & 0 & 0\\
            0 & \epsilon & 0\\
            0 & 0 & \epsilon^2
        \end{bmatrix} \begin{bmatrix}
            a_0\\
            a_0 \bfH\\
            a_0 \bfH^2\\
        \end{bmatrix}.
    \end{align*}
    
\end{example}

\subsection{Bounds on the singular values of $K^{(\ell)}(\bfB,\bfH)$}
\label{sec:sing_val_K}
Let $\mathbf{1}_{s > 0}$ denote the indicator function of the positivity of $s$ and $x$ be a vector with $\norm{x}_2 = 1$. We have
\begin{align*}
    \norm{K^{(\ell)}(\bfB, \bfH) x}^2_2 &= \sum_{j=0}^{\beta-1} \norm{\bfB\bfH^j x}_2^2 + \mathbf{1}_{s>0} \norm{(\bfB \bfH^\beta)_{E_\ell} x}_2^2
    \leq \sigma_{\max}^2(\bfB) \left( \sum_{j=0}^{r_\star} \norm{\bfH}^{2j}\right)
    \leq C_1^2 \sigma_{\max}^2(\bfB),
\end{align*}
where the constant $C_1 = \sqrt{r_\star + 1} \max(1, \norm{\bfH}^{r_\star}_2)$. 

The constant $\tilde{C}_1$ in the upper bound in Corollary \ref{cor:smin} is greater than or equal to $C_1$ owing to the multiplicative factor of the inverse of the finite-difference matrix. Likewise, the constant $c_0$ appears because of the finite-difference matrix and the limiting $K^{(\ell)}(\bfB, \bfH)$.


\subsection{Relation between condition number of the $\bfU_i$'s and $\bfM_j$ matrices}
\label{sec:reln_U_M_cond_no}
In what follows $\bfM_j[i,:]$ denotes the $i$-th block-row of $\bfM_j$ (with $c$ rows) for $i = 0, \dots, N-1$. Based on the structure of $\bfM_j$, we note that for $\norm{x}_2 = 1$
\begin{align*}
    \norm{\bfM_j x}_2^2 &= \sum_{i=0}^{N-1} \norm{\bfM_j[i,:] x}_2^2\\
    &\leq N \max_{i=0, \dots, N-1} \norm{\bfU_i}^2_2\\
    &= N \max_{i=0, \dots, N-1} \sigma^2_{\max}(\bfU_i).
\end{align*}
This means that $\sigma_{\max}(\bfM_j) \leq \sqrt{N} \sqrt{\max_{i=0, \dots, N-1} \sigma^2_{\max}(\bfU_i)}$. Moreover, we also have $\sigma_{\max}(\bfM_j) \geq \max_k \sigma_{\max}(\bfU_k)$. This is because every $\bfU_k$ appears as a zero-padded block row within $\bfM_j$. From Corollary \ref{cor:smin}, we have that $\sigma_{\max}(\bfU_i) = \Theta(\epsilon^{-r_\star})$. Therefore, for fixed $\bfB$ and $\bfH$, we have $\sigma_{\max}(\bfM_j) = \Theta(\epsilon^{-r_\star})$.

Furthermore, for  $\norm{x}_2 = 1$, we have
\begin{align*}
    \norm{\bfM_j x}^2_2 &=  \sum_{i=0}^{N-1} \norm{\bfM_j[i,:] x}^2_2\\
    &= \sum_{i=0}^{N-1} \norm{\bfU_{j+i} x[N-1-i: N+c-2-i]}^2_2\\
    &\geq \left(\sum_{i=0}^{N-1} \norm{x[N-1-i: N+c-2-i]}_2^2\right) \times \left(\min_{k=0, \dots, N-1 } \sigma^2_{\min}(\bfU_k)\right)\\
    &\geq \norm{x}_2^2 \times  \left(\min_{k=0, \dots, N-1 } \sigma^2_{\min}(\bfU_k)\right)\\
    &= \min_{k=0, \dots, N-1 } \sigma^2_{\min}(\bfU_k).
\end{align*}
Thus, we have $\sigma_{\min}(\bfM_j) \geq \sqrt{\left(\min_{k=0, \dots, N-1 } \sigma^2_{\min}(\bfU_k)\right)}$

From Corollary \ref{cor:smin}, we have $\sigma_{\min}(\bfU_i) = \Theta(1)$. Thus, we can conclude that $\kappa_2(\bfM_j) = O(\sqrt{N}\epsilon^{-r_\star})$.

\subsection{Details about numerical experiments}
\label{sec:appendix_num_exp}

{\bf Details about fitting $(\bfU, \bfR)$ parameters for CERAR:}
The code for running the optimization to determine $\bfU$ and $\bfR$ is available at this repository.\footnote{\url{https://anonymous.4open.science/r/CERAR--43C6/}} The code picks $\bfB$ to be a $N \times c$ matrix of i.i.d. $N(0,1)$ random variables.
Following this, the matrix $\bfH$ is either chosen as a diagonal matrix whose diagonal entries are close to zero, or chosen in a deterministic manner using the constructed $\bfB$.

{\bf Details about All-Reduce Phase time measurement experiments:} 
In our experiments, each GPU is controlled by one worker process. Before each measured reduction, the worker processes synchronize at a barrier: each waits until all workers
have reached this point. After synchronization, each worker starts its timer with its GPU
idle, executes the reduction, and stops the timer once the
reduction has finished on its GPU. Before collecting these measurements, we perform 10 untimed warmup runs to
initialize communication resources and, for our algorithm,
compile the GPU functions. The measured time therefore includes communication
and local computation, but excludes warmup runs, input
generation, buffer initialization, synchronization before
the timer starts, and result reporting. We then perform 100 runs measuring
only the total reduction time (unprofiled runs) and 100
additional runs with additional timers measuring the communication and
computation phases separately (profiled runs). For each unprofiled run, we take the maximum
elapsed time across all the GPUs; Table~\ref{tab:reduction_time}
reports the median of these maximum times across runs. 



In the profiled runs, we measure the time spent in our algorithm's
communication and computation phases. For each GPU, we measure
the elapsed time for sending and receiving vectors in each round,
including waiting for these operations to complete. Overlapping
sends and receives are timed together, rather than adding their
individual durations. The communication time is the sum of these
measurements across all rounds. Computation time includes encoding,
updating received vectors and stored data, preparing outgoing
vectors, and decoding, together with the memory accesses required
by these operations.  For the profiling results reported in
Table~\ref{tab:reduction_profile}, we use 1B FP32 entries per GPU on the A100 PCIe system and 24B FP32 entries per GPU on the H200 NVLink system. The table reports
the mean times across the profiled runs, separately for each GPU. For NCCL, we report only the total
profiled reduction time. Communication and computation times
may not sum exactly to the total because the total also includes
overhead outside the measured phases.



\begin{table}[t]
    \centering
    \caption{ {\small Profiled reduction times in milliseconds for A100 PCIe
    (PCIe interconnect; $N=8$, $c=20$, $L=26$, 1B FP32 entries per GPU)
    and H200 SXM (NVLink interconnect; $N=4$, $c=12$, $L=14$,
    24B FP32 entries per GPU).
    Each entry is the mean across profiled runs for the indicated
    GPU. Total is the elapsed time of the complete profiled reduction.}}
    \label{tab:reduction_profile}
    \small
    \setlength{\tabcolsep}{5pt}
    \begin{tabular}{lrrrrr}
        \toprule
        & & NCCL & \multicolumn{3}{c}{Our algorithm} \\
        \cmidrule(lr){3-3}
        \cmidrule(lr){4-6}
        System / interconnect & GPU & Total & Communication & Computation & Total \\
        \midrule
        A100 PCIe, 8 GPUs / PCIe & 0 & 17,927.334 & 9,509.018 & 33.613 &
        9,549.329 \\
        & 1 & 17,924.756 & 9,517.594 & 32.108 & 9,553.995 \\
        & 2 & 17,919.882 & 9,520.844 & 31.662 & 9,556.541 \\
        & 3 & 17,923.538 & 9,519.077 & 32.479 & 9,556.751 \\
        & 4 & 17,927.527 & 9,519.617 & 32.251 & 9,556.561 \\
        & 5 & 17,927.654 & 9,357.713 & 29.141 & 9,387.553 \\
        & 6 & 17,927.335 & 9,289.659 & 31.173 & 9,323.855 \\
        & 7 & 17,927.416 & 9,478.487 & 29.150 & 9,508.829 \\
        \midrule
        H200 SXM, 4 GPUs / NVLink & 0 & 392.158 & 306.509 & 215.889 & 522.592 \\
        & 1 & 392.030 & 306.503 & 215.867 & 522.568 \\
        & 2 & 391.943 & 306.459 & 215.835 & 522.496 \\
        & 3 & 391.948 & 306.591 & 215.729 & 522.512 \\
        \bottomrule
    \end{tabular}
\end{table}

Using the same profiled runs and tensor sizes as
Table~\ref{tab:reduction_profile}, Table~\ref{tab:computation_breakdown} further divides our
algorithm's computation into encoding, outgoing-vector
preparation, updates after reception, and decoding. Thus, the A100 PCIe breakdown corresponds to 1B FP32 entries per GPU, whereas the H200 breakdown corresponds to 24B FP32 entries per GPU. For the A100 PCIe configuration, outgoing-vector preparation covers the additions after
reception in rounds 1--7 that prepare the sends for rounds
2--8. Updates after reception cover rounds 8--26: storing
received vectors for decoding and, except in the final
round, removing the returning local contribution and
preparing the next send. For the H200 configuration, the corresponding round ranges are 1--3 for outgoing-vector preparation (preparing sends for rounds 2--4) and 4--14 for updates after reception. The four components sum to the
computation time in Table~\ref{tab:reduction_profile}, up to
rounding.


\begin{table}[t]
    \centering
    \caption{{\small Computation breakdown for our algorithm in milliseconds,
    using the same profiled runs as Table~\ref{tab:reduction_profile}.
    The A100 PCIe results use a tensor of 1B FP32
    entries per GPU, whereas the H200 SXM results
    use a tensor of 24B FP32 entries per GPU.
    Each entry is the mean for the indicated GPU.}}
    \label{tab:computation_breakdown}
    \small
    \setlength{\tabcolsep}{6pt}
    \begin{tabular}{lrrrrr}
        \toprule
        System / interconnect & GPU & Encoding &
        \shortstack{Outgoing-vector\\preparation} &
        \shortstack{Updates after\\reception} &
        Decoding \\
        \midrule
        A100 PCIe, 8 GPUs / PCIe & 0 & 6.245 & 4.248 & 16.197 & 6.923 \\
        & 1 & 6.186 & 4.120 & 14.907 & 6.896 \\
        & 2 & 6.206 & 3.990 & 14.594 & 6.872 \\
        & 3 & 6.255 & 4.394 & 14.955 & 6.876 \\
        & 4 & 6.186 & 4.351 & 14.842 & 6.873 \\
        & 5 & 6.229 & 3.283 & 12.752 & 6.877 \\
        & 6 & 6.305 & 3.849 & 14.087 & 6.932 \\
        & 7 & 6.222 & 3.508 & 12.544 & 6.876 \\
        \midrule
        H200 SXM, 4 GPUs / NVLink & 0 & 51.390 & 16.441 & 93.414 & 54.645 \\
        & 1 & 51.382 & 16.437 & 93.403 & 54.645 \\
        & 2 & 51.366 & 16.438 & 93.405 & 54.626 \\
        & 3 & 51.386 & 16.422 & 93.276 & 54.645 \\
        \bottomrule
    \end{tabular}
\end{table}

The profiling results help explain the different outcomes.
On A100 PCIe at 1B entries per GPU, our algorithm spends
approximately 9.29--9.52 seconds per GPU on communication,
while its computation takes only 29.1--33.6 ms. On H200 at 24B entries per GPU, our communication phase takes
approximately 306.5--306.6 ms, which is less than NCCL's
complete reduction time of approximately 392.0 ms.
However, the additional computation takes approximately
215.7--215.9 ms, bringing our total reduction time to
approximately 522.5--522.6 ms.
These results demonstrate the benefit of our algorithm on
the tested low-bandwidth PCIe system, while showing that
its computation cost outweighs the communication savings
on the tested high-bandwidth H200 system for our current implementation. As discussed in Section \ref{sec:expts_num}, optimizing our implementation and performing experiments on larger tensors could further improve performance.


{\bf Details about training experiments:} Within each configuration, both methods (CERAR and NCCL) use the same model
initialization, data ordering, and training hyperparameters.
We use sequences of length 512, gradient checkpointing, and the
AdamW optimizer. 
The model computation uses a mixture of BF16 and FP32, while model parameters, stored gradients, and gradient reduction use FP32.
The learning rate is linearly increased to $3\times10^{-4}$ and
subsequently decayed using a cosine schedule. 



For the PCIe experiment, the training set contains 120 million tokens,
each GPU processes a batch of 32 sequences, and the learning-rate
warmup lasts 200 steps. We train for 1,800 steps, processing
117.96 million tokens. For the SXM experiment, the training set
contains approximately 474.89 million tokens, the per-GPU batch size
is 64, and the warmup lasts 100 steps. We train for 7,630 steps,
processing approximately 1 billion tokens.  




\end{document}